\documentclass[
  submission,
  copyright,
]{eptcs}

\providecommand{\event}{GandALF 2026}

\title{Antichains for Concurrent Parameterized Games\thanks{This work
is partially supported by ANR BisoUS (ANR-22-CE48-0012).}}
\def\titlerunning{Antichains for Concurrent Parameterized Games}

\author{
  Nathalie Bertrand
  \institute{Universit\'e de Rennes,
  Inria,\\ CNRS, IRISA, Rennes, France}
  \email{nathalie.bertrand@inria.fr}
  \and
  Patricia Bouyer
  \institute{LMF, CNRS \& ENS Paris-Saclay,\\Universit\'e Paris-Saclay,\\
  Gif-sur-Yvette, France}
  \email{bouyer@lmf.cnrs.fr}
  \and
  Ga\"etan Staquet
  \institute{Nantes Universit\'e,\\\'Ecole Centrale Nantes,\\CNRS, LS2N,
  Nantes, France}
  \email{gaetan.staquet@ec-nantes.fr}
}
\def\authorrunning{N. Bertrand, P. Bouyer, and G. Staquet}

\usepackage{amsthm}
\usepackage{mathtools}
\usepackage{thmtools}

\theoremstyle{plain}
\newtheorem{theorem}{Theorem}[section]
\newtheorem{lemma}[theorem]{Lemma}
\newtheorem{proposition}[theorem]{Proposition}
\newtheorem{corollary}[theorem]{Corollary}

\theoremstyle{definition}
\newtheorem{definition}[theorem]{Definition}
\newtheorem{example}[theorem]{Example}

\usepackage{cleveref}
\usepackage[inline]{enumitem}

\usepackage{xparse}

\hypersetup{
  hidelinks,
}

\usepackage[T1]{fontenc}

\usepackage[dvipsnames,svgnames]{xcolor}
\usepackage{xparse}
\usepackage{xspace}
\usepackage[skins]{tcolorbox}
\usepackage{booktabs}
\usepackage{csquotes}
\usepackage{stmaryrd}
\usepackage{mathtools}
\usepackage{makecell}
\usepackage{tikz}
\usepackage{pgfplots}
\usepackage{amssymb}
\usepackage{multirow}
\usepackage{subcaption}
\usepackage[noadjust,nomove]{cite}

\usetikzlibrary{
  automata,
  arrows.meta,
  calc,
  positioning,
  shapes.geometric,
  patterns,
  plotmarks,
}
\tikzset {
  edge with arrow/.style = {
    ->,
    >=stealth,
    shorten >=1pt,
  },
  directed/.style = {
    edge with arrow,
    node distance=2cm,
    on grid,
    semithick,
    double distance=1.5pt,
  },
  game/.style = {
    directed,
    auto,
    initial text={},
    state/.append style = {
      ellipse,
      inner sep = 1pt,
      minimum size = 10pt,
    },
    eve/.style = {
      state,
      ellipse,
    },
    adam/.style = {
      state,
      rectangle,
      inner sep = 2pt,
    },
    adamStr/.style = {
      ultra thick,
    }
  },
  lattice/.style = {
    on grid,
    semithick,
    auto,
  }
}

\pgfplotsset{
  compat=1.17,
  every x tick label/.append style = {font=\scriptsize},
  every y tick label/.append style = {font=\scriptsize},
  every axis label/.append style = {font=\scriptsize},
}

\NewDocumentCommand{\hairsp}{}{\hspace{1pt}}

\NewDocumentCommand{\ie}{}{\textsl{i.\nobreak\hairsp{}e.}\xspace}
\NewDocumentCommand{\eg}{}{\textsl{e.\nobreak\hairsp{}g.}\xspace}

\RenewDocumentCommand{\implies}{}{\Rightarrow}

\NewDocumentCommand{\nat}{}{\mathbb{N}}
\NewDocumentCommand{\natnonneg}{}{{\nat_{>0}}}
\NewDocumentCommand{\natsubs}{ m }{{\nat_{#1}}}

\NewDocumentCommand{\subsets}{ m }{\mathcal{P}(#1)}

\NewDocumentCommand{\complexity}{ m }{\textsf{#1}}
\NewDocumentCommand{\PTIME}{}{\complexity{PTIME}\xspace}

\NewDocumentCommand{\PSPACE}{}{\complexity{PSPACE}\xspace}

\NewDocumentCommand{\arena}{}{\mathcal{A}}
\NewDocumentCommand{\vertices}{}{V}
\NewDocumentCommand{\alphabet}{}{\Sigma}
\NewDocumentCommand{\transitions}{}{\Delta}
\NewDocumentCommand{\En}{}{\mathrm{En}}
\NewDocumentCommand{\opponents}{m O{}}{\nabla_{#2}(#1)}
\NewDocumentCommand{\histories}{}{\mathsf{Hist}}
\NewDocumentCommand{\strat}{}{\sigma}
\NewDocumentCommand{\outcome}{}{\mathsf{Out}}
\NewDocumentCommand{\target}{}{t}
\NewDocumentCommand{\game}{}{\mathcal{G}}

\NewDocumentCommand{\kGame}{}{\mathsf{Know}_{\game}}
\NewDocumentCommand{\verticesE}{}{{\vertices_E}}
\NewDocumentCommand{\verticesEK}{}{{\vertices_E^\latticeK}}
\NewDocumentCommand{\verticesA}{}{{\vertices_A}}
\NewDocumentCommand{\kTransitions}{}{{\transitions_{\mathsf{Know}}}}
\NewDocumentCommand{\final}{}{F}
\NewDocumentCommand{\kStrat}{}{\lambda}
\NewDocumentCommand{\tree}{}{\mathcal{T}}
\NewDocumentCommand{\Win}{}{\textsf{Win}\xspace}
\NewDocumentCommand{\Lose}{}{\textsf{Lose}\xspace}

\NewDocumentCommand{\parOrder}{}{\sqsubseteq}
\NewDocumentCommand{\parOrdersim}{}{\sqsubseteq_{\mathrm{sim}}}
\NewDocumentCommand{\Dom}{}{\mathsf{Dom}}
\NewDocumentCommand{\reduce}{ m }{\left\lceil #1 \right\rceil}
\NewDocumentCommand{\expand}{ m }{{#1\!\!\downarrow}}
\NewDocumentCommand{\lub}{}{\sqcup}

\NewDocumentCommand{\glb}{}{\sqcap}

\NewDocumentCommand{\W}{}{\mathsf{W}}
\NewDocumentCommand{\WAlt}{}{\mathsf{W}_{\mathrm{Alt}}}
\NewDocumentCommand{\KPred}{m m O{}}{\mathsf{KPred}_{#3}[#1](#2)}
\NewDocumentCommand{\KPredAlt}{m m
O{}}{\mathsf{KPred}^{\mathrm{Alt}}_{#3}[#1](#2)}
\NewDocumentCommand{\Pred}{m O{}}{\mathsf{Pred}_{#2}(#1)}
\NewDocumentCommand{\PredAlt}{m O{}}{\mathsf{Pred}^{\mathrm{Alt}}_{#2}(#1)}
\NewDocumentCommand{\latticeK}{}{\mathcal{K}}

\NewDocumentCommand{\bigsetminus}{}{\mathbin{\big\backslash}}

\AtEndEnvironment{example}{\hfill$\lrcorner$}

\begin{document}
\maketitle

\begin{abstract}
  Concurrent parameterized games involve a fixed yet arbitrary
  number of players.
  They are described by finite arenas in which
  the edges are labeled with languages that describe the possible
  move combinations leading from one vertex to another ($n$ players
  yield a word of length $n$).
  Previous work showed that, when edge labels are regular languages,
  one can decide whether a distinguished
  player, called Eve, has a strategy to ensure a reachability
  objective, against any strategy profile of her arbitrarily many
  opponents. This decision problem is known to be
  \PSPACE-complete. A basic ingredient in the
  \PSPACE-membership proof is
  the reduction to the exponential-size \emph{knowledge game}, a
  2-player game that reflects the knowledge Eve has on the number of
  opponents.

  In this paper, we provide a symbolic approach, based on
  antichains, to compute Eve's winning region in the knowledge
  game. In words, it gives the minimal knowledge Eve needs at every
  vertex to win the concurrent parameterized reachability game. More
  precisely, we propose two fixed-point algorithms that compute, as
  an antichain, the maximal elements of the winning region for Eve
  in the knowledge game. We implemented these two algorithms in C++,
  as well as the one initially proposed, and report on their
  relative performances on various benchmarks.
\end{abstract}

\section{Introduction}\label{sec:intro}

\paragraph{Parameterized verification.} The verification of models
with unknown parameters is known as parameterized
verification. Special cases are systems formed of an arbitrary number
of agents. Motivations for such models are varied: they range from
distributed algorithms and communication protocols to chemical or
biological systems. In all cases, these parameterized systems are
meant to be correct independently of the number of agents.  Verifying
parameterized systems is challenging, since it aims at validating
infinitely many instances at once, independently of the parameterized
number of agents.

Apt and Kozen proved the first undecidability result for the
verification of parameterized systems~\cite{AK-ipl86}. In their
reduction, the system consists of a single agent, whose description
depends on parameter $n$. Oftentimes, when agents bear identities, the
parameterized verification of simple properties is
undecidable. In contrast, in their pioneering work on distributed
models with many
identical entities, German and Sistla represented the behavior of a
network by finite state machines interacting via rendezvous
communications, and provided decision algorithms for their
parameterized verification~\cite{GS-jacm92}. Since then, variants of
models have been proposed to handle various communication
means. Remarkably, subtle changes, such as the presence or absence of
a controller in the system, can drastically alter the complexity of
the verification problems, as surveyed in~\cite{Esparza-stacs14}.

\paragraph{Concurrent parameterized games.}
Concurrent parameterized games~\cite{BertrandBM19} are one of the few
examples of game frameworks in the parameterized verification
landscape. They generalize concurrent 2-player games to an arbitrary number of
players. More precisely, the arena of a concurrent parameterized game
represents infinitely many games, each played with a fixed number of
players. A distinguished player, Eve, aims at achieving her objective
not knowing \emph{a priori} how many opponents she has. Solving such
games amounts to determining whether Eve has a winning strategy whatever
the number of her opponents and their respective strategies. When Eve
wants to achieve a reachability objective, deciding whether she has a
winning strategy, independently of the number of her opponents, is a
\PSPACE-complete problem~\cite{BertrandBM19}. Since Eve's opponents
can collude, the resolution of concurrent parameterized games with
reachability objectives reduces to the resolution of a two-player
game, called the \emph{knowledge game}, where vertices not only store
the original game vertex but also the \emph{knowledge} Eve has on the number
of her opponents. Her adversary in that 2-player game represents the
coalition of her arbitrarily many opponents in the parameterized
game. In the worst-case, the set of possible knowledge --- and hence the
size of the knowledge game --- can be exponential in the size of the
concurrent parameterized game. The DFS algorithm proposed
in~\cite{BertrandBM19} performs an on-the-fly DFS exploration
of the knowledge game.

\paragraph{Contributions.}
We exploit a monotonicity property of the knowledge
game: if Eve wins from a given knowledge set, she also wins from any
finer knowledge set. Building on this observation and the natural
partial order on sets of natural numbers (representing possible number of
opponents), we propose two new algorithms for checking the existence
of a winning strategy for Eve, based on antichains.
A full version of this paper, containing all proofs, appears
at~\cite{BertrandBS25}.

We first characterize the winning positions for Eve in the knowledge
game as the least fixed-point of a monotone function on the lattice of
antichains on knowledges. We turn this non-effective fixed-point
characterization into two algorithms: the first one exploits that only
finitely many knowledge sets are relevant, making it possible to
compute in a finite lattice; the second one takes an operational
viewpoint to compute iteratively the winning region. A non-trivial
contribution consists in proving that these two algorithms indeed
compute the winning region.  Importantly, our new
algorithms do more than the aforementioned DFS algorithm, as the
latter decides whether Eve has a winning strategy from a fixed
initial vertex, while we determine every position from which Eve can
win.

We implemented our two new algorithms and the DFS one proposed
in~\cite{BertrandBM19} in C++ and benchmarked their performances on
several instances of concurrent parameterized games. All instances are
scalable, to demonstrate how the three algorithms scale when instances
grow large (both in number of vertices and relevant knowledge sets).
Inspired by the \PSPACE-hardness proof of~\cite{BertrandBM19}, part of
our benchmarks are concurrent parameterized arenas derived from
quantified boolean formulas (QBF), with up to 7 variables, and a
number of clauses between 50 and 100.
Our benchmarks cover instances of concurrent parameterized games with
up to 126 vertices and hundreds of edges, resulting in a lattice of
2048 to 32768 knowledge sets (per vertex).  We were able to
successfully check the existence or the absence of a winning strategy
for Eve on instances with up to thousands of knowledge vertices in
less than 5 minutes. The experiments show that on instances
of synthetic arenas that are winning for Eve, the running times are
favorable to the
DFS-algorithm. This competitive drawback of our algorithms must be set
against the fact that our algorithms not only decide whether Eve has a
winning strategy, but also compute her winning region: they provide
for every vertex
the minimal information she needs to win.
Conversely, on QBF instances, the algorithm from~\cite{BertrandBM19}
times out for much smaller instances than our antichain
algorithms. Also, there are many negative synthetic instances of
parameterized games, on which the antichain algorithms outperform the
DFS one. Finally, among antichain algorithms, the experiments
highlight the significantly better performance of the one that avoids
the construction of the lattice.

\paragraph{Related work.} Fixed-point computations on lattices were
successfully used in the
verification community to efficiently solve computationally hard
problems that are not \PTIME-tractable: resolution of games of
imperfect information~\cite{DDR-hscc06}, universality of finite
automata~\cite{DDHR-cav06,DR-tacas10}, language inclusion for various
types of automata~\cite{DR-lmcs09,HIRV-fmsd20,DGH-tacas23},
an automata-based approach for LTL
model-checking~\cite{DDMR-tacas08,DR-lmcs09}, LTL
realizability~\cite{FJR-cav09,FJR-fmsd11,CP-tacas23}, QBF
satisfiability~\cite{BBDDR-atva11}, or refinement
checking~\cite{WSSLDWL-icfem12,LGW-lmcs21}.  An antichain is merely a
set of incomparable
elements for a given partial order.  The antichain approach always
exploits some monotonicity of the problem and proposes a fixed-point
(semi-)algorithm iterating a monotone operator over the lattice of
antichains. In these works, the greatest or least fixed-point of the
operator exactly contains the relevant information to solve the
problem. Generally speaking the antichain approach compared favorably,
in terms of experimental complexity, to other approaches from the
literature. The work closest to ours in that series of contributions
is the earliest one, on games of imperfect information. Not only is it
the only one concerned with games, but also, similarly to our setting,
players only have partial information. The two settings however differ
greatly in that in concurrent parameterized games, the information Eve
has is very specific and concerns the number of opponents she has,
whereas in games of imperfect information, the uncertainty is on the
current game vertex.

\section{Concurrent Parameterized Games}\label{sec:setting}

Concurrent parameterized games, introduced in~\cite{BertrandBM19}, are
described by parameterized arenas. In such an arena, edges are labeled
by pairs $(a,k)$, where $a$ is an action of Eve and $k$ a number of
opponents\footnote{Our definition of parameterized arenas specifies
  the number of opponents, and does not reflect their action choices.
  While parameterized arenas were initially defined with edges labeled
  with regular languages, for the problem we are interested in, one
  can focus on the number of opponents without loss of generality
(see~\cite[Section 4]{BertrandBM19}).}. The number of opponents is
initially unknown to Eve, but she may refine her knowledge about $k$
while playing, based on the alternating sequence of actions and
visited vertices. The objective of Eve is to ensure a reachability
objective independently of $k$.

We write $\natnonneg$ for the set of positive integers
and more generally $\natsubs{\gamma}$ for the subset of positive
integers satisfying a constraint $\gamma$, \eg,
$\natsubs{>2}$ for the set of integers larger than $2$.

\begin{definition}[Parameterized arena]\label{def:arena}
  A \emph{parameterized arena} is a tuple
  \(\arena = (\vertices, \alphabet, \transitions)\)
  where
  \(\vertices\) is a finite set of vertices,
  \(\alphabet\) is a finite set of actions, and
  \(\transitions : \vertices \times \alphabet \times \natnonneg \to
  \subsets{\vertices}\) is the transition function.
\end{definition}

The arena is \emph{deterministic} if for every \(v \in \vertices\),
\(a \in \alphabet\), and \(k \in \natnonneg\), there is at most one
vertex \(v' \in \vertices\) such that
\(v' \in \transitions(v, a, k)\).  An action \(a \in \alphabet\) is
\emph{enabled} at vertex \(v\) if there exists \(k \in \natnonneg\)
such that \(\transitions(v, a, k) \neq \emptyset\).  We write
\(\En(v)\) for the set of enabled actions at \(v\).  The arena is
assumed to be \emph{complete for enabled actions}: for every
\(v \in \vertices\), if \(a\) is enabled at \(v\), then for all
\(k \in \natnonneg\), \(\transitions(v, a, k) \neq \emptyset\). This
assumption is natural: Eve does not know how many opponents she has,
and the successor vertex must exist whatever that number is.

For any \(v, v' \in \vertices\) and \(a \in \alphabet\), we introduce
the following notation to represent the set of numbers of opponents that can
lead from \(v\) to \(v'\) under action \(a\) of Eve:
\(\opponents{v, a, v'} = \{k \in \natnonneg \mid v' \in
\transitions(v, a, k)\}\).

For every $v\in \vertices$, for every $a \in \En(v)$ and every
$V' \subseteq V$, we define:
\begin{equation}\label{nablaI}
  \opponents{v,a,V'} = \bigcap_{v' \in V'} \opponents{v,a,v'}
  \bigsetminus
  \bigcup_{v' \notin V'} \opponents{v,a,v'}.
\end{equation}
This corresponds exactly to the numbers of opponents that allow to
reach vertices in $V'$ and forbid to reach vertices outside $V'$.
Note that, on non-deterministic arenas, we may have
\(\opponents{v, a, \{v'\}} \subsetneq \opponents{v, a, v'}\).
As a consequence of the completeness assumption, for every
$v \in \vertices$ and $a \in \En(v)$, for every set
$K \subseteq \natnonneg$:
\begin{equation}\label{K:complete}
  K = \bigcup_{V' \subseteq V} K \cap \opponents{v,a,V'}.
\end{equation}

Although \Cref{def:arena} is quite general and permits
infinitely many outgoing transitions (one for each natural number,
for example), in the rest of the paper, we impose restrictions on
the arenas. Specifically, we assume that all sets
\(\opponents{v, a, v'}\) and \(\opponents{v, a, V'}\) are finite
unions of intervals, as in our implementation
(see~\Cref{sec:experiments}).  However, the developments made
throughout this paper also apply to any family of sets with effective
intersection, union and set difference operations,
provided it yields a finite representation of the
arena, such as semilinear sets for instance.

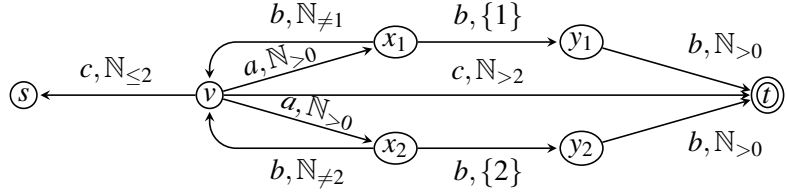
\begin{figure}[t]
  \centering
  \begin{tikzpicture}[
      game,
      node distance = 20pt and 70pt,
    ]
    \node [state]                               (v)   {\(v\)};
    \node [state, above right=of v]             (x1)  {\(x_1\)};
    \node [state, below right=of v]             (x2)  {\(x_2\)};
    \node [state, right=of x1]                  (y1)  {\(y_1\)};
    \node [state, right=of x2]                  (y2)  {\(y_2\)};
    \node [state, accepting, below right=of y1] (t)   {\(t\)};
    \node [state, left=of v]              (bad) {\(s\)};

    \path
    (v) edge node [sloped, pos=0.4, yshift=-3pt]  {\(a, \natnonneg\)}
    (x1)
    edge node [sloped, pos=0.6, yshift=-3pt]  {\(a, \natnonneg\)}
    (x2)
    edge node [sloped]                        {\(c,  \natsubs{\leq
    2}\)}     (bad)
    edge node                                 {\(c, \natsubs{> 2}\)}         (t)
    (x1)  edge node                                 {\(b, \{1\}\)}
    (y1)
    (x2)  edge node [']                             {\(b, \{2\}\)}
    (y2)
    (y1)  edge node                                 {\(b,
    \natnonneg\)}            (t)
    (y2)  edge node [']                             {\(b,
    \natnonneg\)}            (t)
    ;

    \draw [rounded corners = 10pt]
    let
    \p{s} = (x1.west),
    \p{t} = (v.north),
    in
    (\p{s}) -- (\x{t}, \y{s})
    node [above, pos=0.4] {\(b,  \natsubs{\neq 1}\)}
    -- (\p{t})
    ;
    \draw [rounded corners = 10pt]
    let
    \p{s} = (x2.west),
    \p{t} = (v.south),
    in
    (\p{s}) -- (\x{t}, \y{s})
    node [below, pos=0.4] {\(b,  \natsubs{\neq 2}\)}
    -- (\p{t})
    ;
  \end{tikzpicture}
  \caption{An example of a non-deterministic parameterized
  arena.}\label{fig:arena:nondeterministic}
\end{figure}
\begin{example}
  Let \(\arena\) be the parameterized arena of
  \Cref{fig:arena:nondeterministic}.  In our figures, we use
  constraints to represent the transition function: for instance, the
  label \enquote{\(c, \natsubs{\leq 2}\)} on the transition from \(v\)
  to \(s\) represents
  \(\transitions(v, c, 1) = \transitions(v, c, 2) = \{s\}\).

  The label \enquote{\(a, \natnonneg\)} from \(v\) to \(x_1\) denotes
  that \(\opponents{v, a, x_1} = \natnonneg\), while the absence of
  edges from \(v\) to \(y_1\) represents
  \(\opponents{v, a, y_1} = \emptyset\).  Further, as
  \(\opponents{v, a, x_2} = \natnonneg\) and, for any
  \(v' \in \{s, y_1, y_2, t\}\), \(\opponents{v, a, v'} = \emptyset\),
  we have \(\opponents{v, a, \{x_1, x_2\}} = \natnonneg\).  Finally,
  since \(\opponents{x_1, b, v} =  \natsubs{\neq 1}\) and
  \(\opponents{x_1, b, y_1} = \{1\}\), it holds that
  \(\opponents{x_1, b, \{v, y_1\}} = \emptyset\).
  Note that this arena
  is complete for enabled actions.
\end{example}

Let \(k \in \natnonneg\).
A \emph{\(k\)-history}, for a coalition composed of \(k\) opponents of Eve, is
a finite sequence
\(v_0 a_0 \dotsb v_i \in {(\vertices \cdot \alphabet)}^* \cdot \vertices\)
such that for every \(j < i, v_{j+1} \in \transitions(v_j, a_j, k)\).
A \emph{history} in \(\arena\) is a \(k\)-history for some \(k \in \natnonneg\).
We write \(\histories(k)\) (resp.\ \(\histories\)) for the set of
\(k\)-histories
(resp.\ histories) in \(\arena\).
Similar notions of a \emph{\(k\)-play} and a \emph{play} are defined for
infinite sequences.

\begin{definition}[Strategy]
  A \emph{strategy} for Eve in \(\arena\) is a mapping
  \(\strat : \histories \to \alphabet\) that associates to every
  history \(hv' \in \histories\) an action \(\strat(hv') \in \En(v')\).
\end{definition}

A strategy for Eve is applied with no prior information on the number
of opponents. The resulting plays however may highly
depend on it. Given a strategy \(\strat\), an initial vertex \(v\)
and \(k \in \natnonneg\) a number of opponents, we define the outcome
\(\outcome(\strat, v, k)\) as the set of plays that \(\strat\) induces
from \(v\) when Eve has exactly \(k\) opponents:
\(\outcome(\strat, v, k)\) is the set of all \(k\)-plays
\(\rho = v_0 a_0 v_1 a_1 v_2 \dotsb\) such that \(v = v_0\), and for
all \(i \geq 0\), \(\strat(v_0 a_0 \dotsb v_i) = a_i\) and
\(v_{i+1} \in \transitions(v_i, a_i, k)\).  The completeness
assumption ensures that the set \(\outcome(\strat, v, k)\) is not
empty.  Finally, \(\outcome(\strat, v)\) is the set of all possible
plays induced by \(\strat\) from \(v\):
\(\outcome(\strat, v) = \bigcup_{k \geq 1} \outcome(\strat, v, k)\).

Given a parameterized arena
\(\arena = (\vertices, \alphabet, \transitions)\), a target vertex
\(\target \in \vertices\) defines a \emph{reachability game}
\(\game = (\arena, \target)\) for Eve.  A strategy \(\strat\) for Eve
from \(v\) in $\game$ is
\emph{winning} if all plays in \(\outcome(\strat, v)\) eventually
reach \(\target\).  In this case,
\(v\) belongs to the \emph{winning region} of Eve.
On the arena of \Cref{fig:arena:nondeterministic}, in which the target
vertex is highlighted by a double circle, the winning region of Eve is
$\{t, y_1,y_2\}$, as argued in \Cref{ex:knowledgeGame} below.

For constraints $\nabla(v,a,v')$ that are finite unions of intervals or
semilinear sets, the paper~\cite{BertrandBM19} studies the existence
of a winning strategy for Eve, and shows:
\begin{theorem}[\cite{BertrandBM19}]
  Deciding the existence of a winning strategy for Eve is \PSPACE-complete.
\end{theorem}

The purpose of the current paper is to use lattices and antichains to
compute the winning region of Eve.
We first recap some results of~\cite{BertrandBM19}.

\paragraph{Knowledge game.}

The DFS algorithm of~\cite{BertrandBM19} first turns a
parameterized reachability game into a standard two-player turn-based
game~\cite[Chapter 2]{GamesOnGraphs}.

\begin{definition}[Knowledge game~\cite{BertrandBM19}]\label{def:knowledge}
  Let \(\game = ((\vertices, \alphabet, \transitions), \target)\) be a
  concurrent parameterized reachability game. The \emph{knowledge
  game} associated with \(\game\) is the two-player turn-based
  reachability game
  \(\kGame = (\verticesE \cup \verticesA, \kTransitions, \final)\)
  between Eve and Adam, such that
  \(\verticesE = \vertices \times \subsets{\natnonneg}\) and
  \(\verticesA = \verticesE \times \alphabet\) are Eve and Adam
  vertices, respectively;
  \(\final = \{(\target, K) \mid K \subseteq \natnonneg\}\) is the set
  of target vertices; and
  \(\kTransitions \subseteq (\verticesE \times \verticesA) \cup
  (\verticesA \times \verticesE)\) is the edge relation defined as follows:
  \begin{itemize}
    \item
      for all \((v, K) \in \verticesE\) and \(a \in \alphabet\)
      enabled at \(v\),
      \(((v, K), (v, K, a)) \in \kTransitions\);
    \item
      for all \((v, K, a) \in \verticesA\) and \(v' \in \vertices\),
      \(((v, K, a), (v', K \cap \opponents{v, a, v'})) \in \kTransitions\)
      whenever \(K \cap \opponents{v, a, v'} \neq \emptyset\).
  \end{itemize}
\end{definition}

Intuitively, each vertex \((v, K)\) owned by Eve records
her current knowledge on the number of her opponents.  We call the
set \(K\) a \emph{knowledge set}.  When she selects a successor
\((v, K, a)\), Adam picks a vertex
\((v', K \cap \opponents{v, a, v'})\) and the knowledge of Eve is updated:
if she learns that the number of opponents cannot be
some \(k\), she retains this information forever, and can use it when
selecting a later edge.
Notice that the reachable part
of \(\kGame\) from any \((v, \natnonneg)\) is finite: one can show that
for every vertex \((v, K)\) that Eve can reach, \(K\) is the
intersection of finitely many sets of the form
\(\opponents{v', a, v''}\), or \(\natnonneg\).

A \emph{strategy} for Eve in \(\kGame\) is a mapping
\(\kStrat : {(\verticesE \cdot \verticesA)}^* \cdot \verticesE \to \verticesA\)
compatible with \(\kTransitions\).
We borrow standard notions of outcomes and winning strategies from
the literature.

\begin{theorem}[\cite{BertrandBM19}]\label{thm:knowledge:strategy}
  Eve has a winning strategy
  from \(v\) in \(\game\) if and only if she has a winning
  strategy from \((v, \natnonneg)\) in \(\kGame\).
\end{theorem}

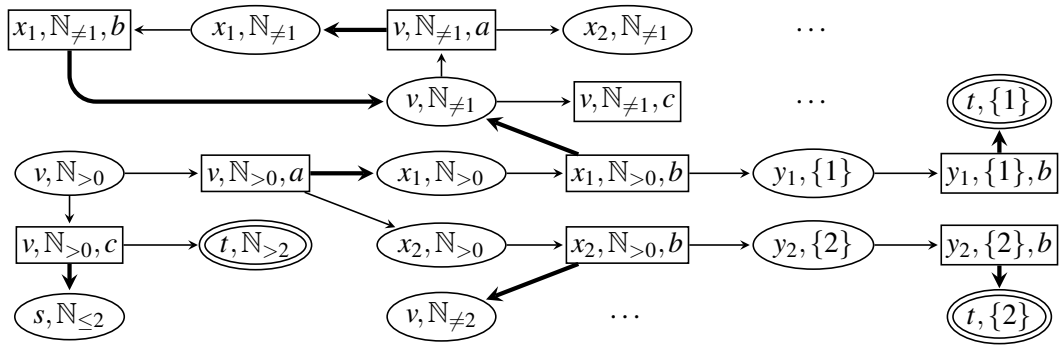
\begin{figure}[t]
  \centering
  \begin{tikzpicture}[
      game,
      node distance = 27pt and 70pt,
    ]
    \node [eve]                             (v N)       {\(v, \natnonneg\)};
    \node [adam, below=of v N]              (v N c)     {\(v, \natnonneg, c\)};
    \node [adam, right=of v N]              (v N a)     {\(v, \natnonneg, a\)};
    \node [eve, below=of v N c]              (s leq 2)  {\(s,
    \natsubs{\leq 2}\)};
    \node [eve, accepting, right=of v N c]  (t gt 2)    {\(t,  \natsubs{> 2}\)};
    \node [eve, right=of v N a]             (x1 N)      {\(x_1, \natnonneg\)};
    \node [eve, below=of x1 N]              (x2 N)      {\(x_2, \natnonneg\)};
    \node [adam, right=of x1 N]             (x1 N b)    {\(x_1,
    \natnonneg, b\)};
    \node [adam, right=of x2 N]             (x2 N b)    {\(x_2,
    \natnonneg, b\)};
    \node [eve, above=of x1 N]              (v neq 1)   {\(v,
    \natsubs{\neq 1}\)};
    \node [eve, below=of x2 N]              (v neq 2)   {\(v,
    \natsubs{\neq 2}\)};
    \node [right=of v neq 2]                            {\ldots};
    \node [eve, right=of x1 N b]            (y1 1)      {\(y_1, \{1\}\)};
    \node [eve, right=of x2 N b]            (y2 2)      {\(y_2, \{2\}\)};
    \node [adam, right=of y1 1]             (y1 1 b)    {\(y_1, \{1\}, b\)};
    \node [adam, right=of y2 2]             (y2 2 b)    {\(y_2, \{2\}, b\)};
    \node [eve, accepting, above=of y1 1 b] (t 1)       {\(t, \{1\}\)};
    \node [eve, accepting, below=of y2 2 b] (t 2)       {\(t, \{2\}\)};
    \node [adam, right=of v neq 1]          (v neq 1 c) {\(v,
    \natsubs{\neq 1}, c\)};
    \node [right=of v neq 1 c]                          {\ldots};
    \node [adam, above=of v neq 1]          (v neq 1 a) {\(v,
    \natsubs{\neq 1}, a\)};
    \node [eve, left=of v neq 1 a]          (x1 neq 1)  {\(x_1,
    \natsubs{\neq 1}\)};
    \node [eve, right=of v neq 1 a]         (x2 neq 1)  {\(x_2,
    \natsubs{\neq 1}\)};
    \node [right=of x2 neq 1]                           {\ldots};
    \node [adam, left=of x1 neq 1]          (x1 neq 1 a){\(x_1,
    \natsubs{\neq 1}, b\)};

    \foreach \s/\t in {v N/v N a,v N/v N c,
      v N c/t gt 2,
      v N a/x2 N,
      x1 N/x1 N b,
      x2 N/x2 N b,
      x1 N b/y1 1,
      x2 N b/y2 2,
      y1 1/y1 1 b,
      y2 2/y2 2 b,
      v neq 1/v neq 1 c,v neq 1/v neq 1 a,
      v neq 1 a/x2 neq 1,
    x1 neq 1/x1 neq 1 a} {
      \path (\s) edge (\t);
    }

    \foreach \s/\t in {v N a/x1 N,
      x1 N b/v neq 1,
      v neq 1 a/x1 neq 1,
      x2 N b/v neq 2,
      y1 1 b/t 1,
      y2 2 b/t 2,
    v N c/s leq 2} {
      \path (\s) edge [adamStr] (\t);
    }

    \draw [rounded corners = 10pt, adamStr]
    let
    \p{s} = (x1 neq 1 a.south),
    \p{t} = (v neq 1.west),
    in
    (\p{s}) -- (\x{s}, \y{t})
    -- (\p{t})
    ;
  \end{tikzpicture}
  \caption{Fragment of the knowledge game from \(v\) for the arena of
  \Cref{fig:arena:nondeterministic}.}\label{fig:knowledgeGame}
\end{figure}
\begin{example}\label{ex:knowledgeGame}
  \Cref{fig:knowledgeGame} represents a fragment of the knowledge
  game associated
  with the parameterized arena of \Cref{fig:arena:nondeterministic}.
  Oval vertices belong to Eve, and rectangular ones to Adam.
  In this two-player game, Eve does not have a winning strategy
  from \((v, \natnonneg)\).
  A memoryless counterstrategy for Adam is drawn in bold in the
  figure. In particular, if Eve chooses to go to
  \((v, \natnonneg, a)\), Adam's strategy forces a cycle in the game
  that never visits \(t\).
  However, for any \(K \subseteq \natnonneg\), Eve has a winning
  strategy from \((y_1, K)\) and from \((y_2, K)\). In particular, she
  wins from \(y_1\) and from \(y_2\) in the parameterized game, thanks
  to \Cref{thm:knowledge:strategy}.
\end{example}

\section{Symbolic algorithms for parameterized games}\label{sec:symbolic}

Let \(\game\) be a parameterized reachability game, and \(\kGame\) be
its associated knowledge game.  Our objective is to compute Eve's
winning region in \(\kGame\).  Thanks to
\Cref{thm:knowledge:strategy}, deriving Eve's winning region in
\(\game\) will then be straightforward.  Instead of storing every
vertex of the winning region in \(\kGame\), we obtain a more compact
representation by leveraging \emph{antichains}.  We first formally
define antichains and their properties, before specializing them to
represent vertices in \(\kGame\).  We then give a fixed-point
characterization of the winning region via antichains and show its
correctness, before providing a different, more operational
fixed-point approach. We finally prove the non-trivial
fact that both algorithms compute the same sets, \ie, they both
output Eve's winning region in \(\kGame\).

\subsection{A summary of the theory of antichains}\label{subsec:antichains}

Before introducing antichains (in short, sets of incomparable elements according
to some partial order), let us focus on partially ordered sets
and (complete) lattices.

\subparagraph{Partially ordered sets.}

Formally, for \(X\) a set, a \emph{partial order} on $X$
is a binary relation \({\parOrder} \subseteq X \times X\) that is
reflexive, transitive, and antisymmetric (\ie, for any \(a, b \in X\),
if \(a \parOrder b\) and \(b \parOrder a\), then \(a = b\)).  The pair
\((X, \parOrder)\) denotes a \emph{partially ordered set}.
We say that \((X, \parOrder)\) is a \emph{lattice} if every
pair of elements \(x \neq y\) of \(X\) has a
\emph{join} (a least upper bound) and a \emph{meet} (a greatest lower bound)
that both belong to \(X\).

\begin{example}\label{ex:lattice}
  The partially ordered set \((\subsets{\nat}, \subseteq)\) is a
  lattice, with the join defined as the union, and the meet as the
  intersection. For instance, the join of \(\{1, 3\}\) and
  \(\{2, 3\}\) is \(\{1, 2, 3\}\), while their meet is \(\{3\}\). As
  both elements belong to \(\subsets{\nat}\) (and this can be easily
  generalized to any pair of sets), \((\subsets{\nat}, \subseteq)\) is
  indeed a lattice.

  As another example, let us consider the partially ordered set
  \((\{1, 2, 3, 6\}, \parOrder)\) such that, for any
  \(a, b \in \nat\), \(a \parOrder b\) if and only if \(a\) divides
  \(b\).  Then, by defining the join as the least common multiple
  (\eg, the join of 2 and 3 is 6) and the meet as the greatest common
  divisor (\eg, the meet of 2 and 3 is 1),
  \((\{1, 2, 3, 6\}, \parOrder)\) forms a lattice. However, the
  partially ordered set \((\{1, 2, 3\}, \parOrder)\) does not form a
  lattice, as the join of 2 and 3 is not defined within
  \(\{1, 2, 3\}\).
\end{example}

\subparagraph{Complete lattices.}
Throughout this paper, we will rely on special lattices: a partially ordered
set \((X, \parOrder)\) is a \emph{complete lattice} if, for every
subset \(L \subseteq X\), \(L\) has both a join and a meet in
\((X, \parOrder)\). Notice that any complete lattice is necessarily a lattice.

\begin{example}\label{ex:lattice:complete}
  The lattice \((\subsets{\nat}, \subseteq)\) is actually a complete
  lattice.  Indeed, for any \(L \subseteq \subsets{\nat}\), let join
  be the union and meet be the intersection.  For instance, the join
  of \(\{\{1, 3\}, \{2, 3\}, \{3, 4\}\}\) is \(\{1, 2, 3, 4\}\), and
  its meet is \(\{3\}\).

  The lattice \((\{1, 2, 3, 6\}, \parOrder)\) from the previous
  example is also a complete lattice. Indeed, since there are finitely
  many elements in this lattice, any subset must also be finite and,
  thus, has well-defined join and meet. This holds for any non-empty
  finite lattice.  Naturally, since \((\{1, 2, 3\}, \parOrder)\) is
  not a lattice, it cannot be a complete lattice.
\end{example}

\begin{proposition}\label{prop:lattice:finite:complete}
  Any non-empty finite lattice is a complete lattice.
\end{proposition}

\subparagraph{Antichains.}

Our presentation of antichains, a central tool in our work,
follows~\cite{DDR-hscc06}.  Intuitively, antichains are sets of
incomparable elements of a partially ordered set \((X, \parOrder)\).
An \emph{antichain} for \((X, \parOrder)\) is a subset
\(X' \subseteq X\) such that any two elements of \(X'\) are
incomparable for \(\parOrder\): \(\forall a \neq b \in X'\),
\(a \not\parOrder b\) and \(b \not\parOrder a\).

Let \(L \subseteq X\) be a set. An antichain can be used to represent or
abstract \(L\) by keeping only the maximal elements for \(\parOrder\).
We say that an element \(x \in X\) is \emph{dominated} by \(L\) if
there exists some \(x' \in L\) such that \(x \parOrder x'\) and
\(x \neq x'\).  We write \(\Dom(L)\) for the set of all elements of
\(X\) that are dominated by \(L\):
\(\Dom(L) = \{x \in X \mid \exists x' \in L : x \parOrder x' \land x
\neq x'\}\).

The \emph{reduced form} of \(L\), noted \(\reduce{L}\), is obtained by
removing from \(L\) all its dominated elements, while the
\emph{expanded form} of \(L\), noted \(\expand{L}\), is obtained by
adding to $L$ all dominated elements in $X$.  Formally,
\( \reduce{L} = L \setminus \Dom(L) \) and
\( \expand{L} = L \cup \Dom(L)\).  A set $L$ such that
$L = \expand{L}$ is said \emph{downward-closed}.  As
\(\Dom(L) \cap \reduce{L} = \emptyset\) (\ie, \(\reduce{L}\) does not
contain any element that is dominated by \(L\)), any two distinct
elements of \(\reduce{L}\) are incomparable for \(\parOrder\), and
\(\reduce{L}\) is an antichain.  Notice that, for any \emph{finite}
non-empty subset \(L\) of \(X\), there is at least one element that is
not dominated by \(L\). That is, \(L \nsubseteq \Dom(L)\), meaning
that \(\reduce{L}\) is not empty.  Moreover,
\(\expand{L} = \expand{\reduce{L}}\). This does not always hold for
infinite sets: for instance, the subset
\(L = \{\{1\}, \{1, 2\}, \dotsc\}\) of \((\subsets{\nat}, \subseteq)\)
does not have a maximal element, \ie, \(\reduce{L} =
\emptyset\). Among other results, the next proposition gives a
necessary and sufficient condition for
\(\expand{L} = \expand{\reduce{L}}\) to hold, even for infinite
sets. Antichains are particularly adapted to representing such
downward-closed sets \(\expand{L}\), as they are uniquely determined
from the antichain $\reduce{L}$.

\begin{restatable}{proposition}{propExpandReduce}\label{prop:expand_is_expand_reduce}
  For any sets \(L, L' \subseteq X\), the following facts all hold.
  \begin{enumerate}
    \item
      \(\reduce{L}\) is an antichain for
      \((X, \parOrder)\).
    \item
      \(\Dom(L) = \Dom(\expand{L})\).
    \item
      \(\expand{L} \cup \expand{L'} = \expand{(L \cup L')}\).
    \item
      \(\reduce{L} = \reduce{\expand{L}}\).
    \item
      \(\expand{L} = \expand{\reduce{L}} \iff \forall x\in L, \exists
      y \in \reduce{L} : x \parOrder y\).
    \item
      \(\expand{L} = \expand{\reduce{L}} \land
        \expand{L'} = \expand{\reduce{L'}}
        \implies
        \expand{\reduce{L \cup L'}} = \expand{(\reduce{L} \cup
      \reduce{L'})}\).
  \end{enumerate}
\end{restatable}

\subparagraph{Join and meet for antichains.}
Let us argue that the set of antichains of \((X, \parOrder)\) is a
complete lattice.
To that end, we define the join and meet of two sets \(L, L'
\subseteq X\) as follows.
\begin{description}
  \item[Join] An element \(x\) in
    \(\expand{L} \cup \expand{L'}\) is maximal when it is maximal in
    one of the two sets \(L\) or \(L'\). Hence, we can define the
    \emph{join} operation of \(L\) and \(L'\) as
    \(L \lub L' = \reduce{\expand{L} \cup \expand{L'}}\).
    Since \(\expand{L} \cup \expand{L'} = \expand{(L \cup L')}\),
    \(L \lub L' = \reduce{\expand{L} \cup \expand{L'}} =
    \reduce{\expand{(L \cup L')}} = \reduce{L \cup L'}\), by
    \Cref{prop:expand_is_expand_reduce}.
  \item[Meet] The \emph{meet}
    operation of \(L\) and \(L'\) is
    \(L \glb L' = \reduce{\expand{L} \cap \expand{L'}}\). It
    does not coincide in general with $\reduce{L \cap L'}$.
\end{description}

We emphasize that \(L \lub L'\) and \(L \glb L'\) are antichains
for \((X, \parOrder)\).  The join operation provides the least upper
bound, while the meet operation gives the greatest lower bound
for the two sets \(L\) and \(L'\).
Thus, the set of all antichains of \((X, \parOrder)\) forms a
\emph{complete lattice}.
When \(\expand{L} = \expand{\reduce{L}}\) and
\(\expand{L'} = \expand{\reduce{L'}}\), one can show that computing
\(L \lub L'\) is equivalent to computing the join of the maximal
elements of \(L\) and \(L'\).

\begin{proposition}\label{prop:lub:reduceLub}
  Let \((X, \parOrder)\) be a partially ordered set, and
  \(L, L' \subseteq X\) be two sets such that \(\expand{L} =
  \expand{\reduce{L}}\) and \(\expand{L'} = \expand{\reduce{L'}}\).
  Then,
  \(L \lub L' = \reduce{L} \lub \reduce{L'}\).
\end{proposition}

Finally, we lift $\parOrder$ to subsets of $X$: for two
subsets $L,L' \subseteq X$, we write $L \parOrdersim L'$ whenever for
every $a \in L$, there is $a' \in L'$ such that $a \parOrder a'$. In
particular,
$L \subseteq L'$ implies $L \parOrdersim L'$,
$\reduce{L} \parOrdersim L \parOrdersim \expand{L}$, and
$\expand{L} \parOrdersim L \parOrdersim \reduce{L}$.

The following lemma
relates several of the above-defined notions, and will be useful in
the sequel:
\begin{restatable}{lemma}{lemmareduce}\label{lemma:reduce}
  For any partially ordered set \((X, \parOrder)\)
  and \(L, L' \subseteq X\) such that \(\expand{L} =
  \expand{\reduce{L}}\) and \(\expand{L'} = \expand{\reduce{L'}}\), the three
  following properties are
  equivalent:\[
    \mathit{1.}\ \reduce{L} = \reduce{L'}; \qquad \mathit{2.}\ L
    \parOrdersim L'\ \text{and}\ L'
    \parOrdersim L; \qquad \mathit{3.}\ L \subseteq
    \expand{L'}\ \text{and}\ L' \subseteq \expand{L}.
  \]
\end{restatable}

\subsection{Lattices and antichains for the vertices of
Eve}\label{subsec:appli-antichains}

Aiming at using antichains to compute the winning region for Eve in
\(\kGame\), we make the following observation: if Eve has a winning
strategy from some \((v, K) \in \verticesE\), then she also wins from
any \((v, K')\) such that \(K' \subseteq K\).  In words, refining the
knowledge on the number of opponents can only help Eve to reach a
target vertex. Therefore, her winning region is downward-closed for a
natural partial order on her vertices in $\kGame$, which we now define.

Given a set \(S\), we write \(\subsets{S}\) for the set of all subsets
of \(S\).  Now, one naturally defines the following partial order on
vertices of Eve in the knowledge game: for any
\((v, K), (v', K') \in \verticesE\), \((v, K) \parOrder (v', K')\) if
\(v = v'\) and \(K \subseteq K'\). Equipped with this partial order,
$(\verticesE,\parOrder)$ is a complete lattice (as
  \((\subsets{\natnonneg}, \subseteq)\) is a complete
lattice), which will be
used to compute the winning region for Eve in \(\kGame\).

We use the notation \(\expand{\cdot}\) for downward closure both in
\(\subsets{\natnonneg}\) and \(\verticesE\) (no
  confusion should arise since elements of these two lattices are
``typed'' differently). Recall that $\reduce{\cdot}$
denotes the maximal elements.

\begin{example}
  Let \(\arena\) be the arena of \Cref{fig:arena:nondeterministic} and
  \(L\) be
  \(\{(t, \natnonneg), (y_1, \{1\}),\allowbreak (y_1, \{1, 2\}),
  (y_2, \{1, 2\})\}\).
  The pair \((y_1, \{1\})\) is dominated by \(L\), as \((y_1, \{1, 2\}) \in L\),
  \((y_1, \{1\}) \parOrder (y_1, \{1, 2\})\), and
  \((y_1, \{1\}) \neq (y_1, \{1, 2\})\).
  Moreover, any \((t, K)\) with \(K \subsetneq \natnonneg\) is dominated by
  \(L\).
  However, \((y_1, \{1, 2\})\) is not dominated, as the only element \((v, K)\)
  of \(L\) such that \((y_1, \{1, 2\}) \parOrder (v, K)\) is equal to
  \((y_1, \{1, 2\})\).
  Likewise, \((t, \natnonneg)\) and \((y_2, \{1, 2\})\) are not dominated.
  Thus,
  \begin{equation*}
    \Dom(L) = \{(t, K) \mid K \subsetneq \natnonneg\}
    \cup
    \{(y_1, \emptyset), (y_1, \{1\}), (y_1, \{2\}),
    (y_2, \emptyset), (y_2, \{1\}), (y_2, \{2\})\},
  \end{equation*}
  and \(\reduce{L} = L \setminus \Dom(L) =
  \{(t, \natnonneg), (y_1, \{1, 2\}), (y_2, \{1, 2\})\}\).
\end{example}

\subsection{Fixed-point characterization of the winning
region}\label{subsec:symbolic-algo}
Similarly to the case of standard two-player turn-based reachability
games, the winning region of Eve in the knowledge game
is the least fixed-point of an analogue of the
constrained-predecessor
operator.  Since
\(\final = \{(\target, K) \mid K \subseteq \natnonneg\}\), the
fixed-point is initialized with \((\target, \natnonneg)\), as well as
all \((v, \emptyset)\), for every
\(v \in \vertices \setminus \{\target\}\).
Indeed, \((v, \emptyset)\) represents Eve's current
knowledge on the numbers of opponents she can win against from \(v\):
a priori, she always loses.
The operator \(\Pred{\cdot}\) then reflects the edges in \(\kGame\) by
incorporating all pairs \((v, K)\) such that for some action $a$
enabled at $v$, \((v', K \cap \opponents{v, a, v'})\) is already known
to be winning for every successor vertex $v'$.

\begin{definition}\label{def:symbolic}\label{def:fixpoint}
  For $(v, a) \in \vertices \times \alphabet$ and \(S \subseteq
  \verticesE\), let
  \(\KPred{v,a}{S} \) be the operator on $\subsets{\natnonneg}$,
  and let \(\Pred{S}\) be the operator on \(\subsets{\verticesE}\)
  defined as
  \begin{gather*}
    \KPred{v,a}{S} = \{K \in \subsets{\natnonneg} \mid \forall v' \in
    \vertices : (v', K \cap \opponents{v, a, v'}) \in S\}
    \\
    \Pred{S} = \bigcup_{v\in \vertices} \bigcup_{a\in \En(v)}\{(v,K)
    \mid K \in \KPred{v,a}{S}\}.
  \end{gather*}
  Finally, the sets \({(\W^i)}_{i \in \nat} \subseteq \verticesE\) are
  defined inductively:
  \begin{align*}
    \W^0 &= \{(\target, \natnonneg)\} \cup
    \{(v, \emptyset) \mid v \in \vertices \setminus \{\target\}\}
    \\
    \forall i \geq 0 : \W^{i+1} &= \W^i \lub
    \Pred{\expand{\W^i}} \quad\quad
    \text{\small (computation made in
    lattice $(\verticesE,\parOrder)$)}.
  \end{align*}
\end{definition}
The sets ${(\W^i)}_{i \in \nat}$ are well-defined as subsets of
$\verticesE$. By definition of the join operation,
we only keep the maximal elements with respect to \(\parOrder\):
each \(\W^i\) is an antichain of \((\verticesE, \parOrder)\).
Furthermore, we only consider actions \(a\) that are enabled at \(v\)
to ensure that there exists \(v'\) such that \(\opponents{v,a,v'}\) is
not empty.
Interestingly, there is no need to construct \(\kGame\) to define the
iterations of \(\W^i\).

\begin{example}\label{ex:WK}
  Let \(\game = (\arena, t)\), with
  \(\arena\) the arena of \Cref{fig:arena:nondeterministic}.
  We illustrate the approach with the full lattice \(\subsets{\natnonneg}\).
  The initial set is
  \[
    \W^0 = \{(s, \emptyset), (v, \emptyset), (x_1, \emptyset),
      (x_2, \emptyset),
    (y_1, \emptyset), (y_2, \emptyset), (t, \natnonneg)\}.
  \]
  Now, \(\KPred{v, c}{\expand{\W^0}}\) contains \(\natsubs{>2}\), as
  \((t, \natsubs{>2}) \in \expand{\W^0}\) and, for every \(v' \neq t\),
  \((v', \emptyset) \in \expand{\W^0}\).  Similarly,
  \(\KPred{v, c}{\expand{\W^0}}\) contains
  \(\{3\}, \{4\}, \dotsc, \{3, 4\},
  \dotsc, \allowbreak{\natsubs{>3}, \natsubs{>4},} \dotsc\)
  which are all subsets of \(\natsubs{>2}\).  Recall that
  \(\W^0 \sqcup \Pred{\expand{\W^0}}\) only contains the non-dominated
  elements, \ie, only \(\natsubs{>2}\) in this case.  Moreover,
  \(\KPred{y_1, b}{\expand{\W^0}} = \KPred{y_2, b}{\expand{\W^0}}\)
  contains \(\natnonneg\) and, thus, every subset of \(\natnonneg\).
  Yet, we have
  \(\KPred{x_1, b}{\expand{\W^0}} = \KPred{x_2, b}{\expand{\W^0}} =
  \{\emptyset\}\), as \((v, \emptyset), (y_1, \emptyset)\), and
  \((y_2, \emptyset)\) all belong to \(\W^0\). So,
  \[
    \W^1 = \{(s, \emptyset), (v, \natsubs{>2}), (x_1,
      \emptyset), (x_2, \emptyset), (y_1, \natnonneg), (y_2,
    \natnonneg), (t, \natnonneg)\}.
  \]
  Observe that, at each iteration, it is sufficient to consider only
  the predecessors of the vertices whose knowledge set changed, \ie,
  only \(x_1\) and \(x_2\) in this example.  We focus on \(x_1\), as
  the computations for
  \(x_2\) are similar:
  \(\KPred{x_1, b}{\expand{\W^1}}\) does not contain \(\natnonneg\),
  as \((v, \natsubs{\neq 1}) \notin \expand{\W^1}\).  However,
  \(\natsubs{\neq 2} \in \KPred{x_1, b}{\expand{\W^1}}\), as
  \((v, \natsubs{\neq 2} \cap \natsubs{\neq 1}) = (v,
  \natsubs{>2}) \in \expand{\W^1}\) and
  \((y_1, \natsubs{\neq 2} \cap \{1\}) = (y_1, \{1\}) \in
  \expand{\W^1}\).  Hence,
  \[
    \W^2 = \{(s, \emptyset), (v, \natsubs{>2}), (x_1, \natsubs{\neq
      2}), (x_2, \natsubs{\neq 1} ), (y_1, \natnonneg), (y_2,
    \natnonneg), (t, \natnonneg)\}.
  \]
  As only the knowledge sets of \(x_1\) and \(x_2\) are different, it
  suffices to compute \(\KPred{v, a}{\expand{\W^2}}\). Doing so, one
  gets \(\W^3 = \W^2\), hence
  \({(\expand{\W^i})}_{i \in \natnonneg}\) stabilizes.
\end{example}

Observe that
\(\reduce{\W^i} = \reduce{\W^i \lub \Pred{\expand{\W^i}}}
= \W^i \lub \Pred{\expand{\W^i}} = \W^i\), as \(\lub\) already
keeps only the maximal elements.
By \Cref{prop:expand_is_expand_reduce}, \(\W^i = \reduce{\W^i} =
\reduce{\expand{\W^i}}\).
We argue that the sequence
\({(\expand{\W^i})}_{i \in \nat}\) is non-decreasing. To that end, we
show the following lemma, allowing us to apply
\Cref{prop:lub:reduceLub}.

\begin{lemma}\label{lemma:W:welldefined}
  For any \(i \in \nat\), \(\expand{\Pred{\expand{\W^i}}} =
  \expand{\reduce{\Pred{\expand{\W^i}}}}\).
\end{lemma}
\begin{proof}
  By \Cref{prop:expand_is_expand_reduce}, we show that, for any
  \((v, K) \in \Pred{\expand{\W^i}}\), there exists
  \((v', K') \in \reduce{\Pred{\expand{\W^i}}}\) such that
  \((v, K) \parOrder (v', K')\).
  By definition, it is only possible when \(v = v'\). We can thus rephrase
  it as:
  \begin{equation}
    \forall (v, K) \in \Pred{\expand{\W^i}},
    \exists (v, K') \in \reduce{\Pred{\expand{\W^i}}} :
    K \subseteq K'.
    \label{eq:lemma:W:welldefined:goal}
  \end{equation}

  Let \((v, K_1) \in \Pred{\expand{\W^i}}\). If \((v, K_1) \in
  \reduce{\Pred{\expand{\W^i}}}\), we immediately get the result,
  with \(K' = K_1\).
  Assume \((v, K_1) \notin \reduce{\Pred{\expand{\W^i}}}\), \ie,
  \((v, K_1) \in \Dom(\Pred{\expand{\W^i}})\).
  So, there exists some \((v, K_2) \in
  \Pred{\expand{\W^i}}\) such that \(K_1 \subsetneq K_2\). Again, if
  \((v, K_2)\) is a maximal element, we are done. Otherwise, there is
  some \((v, K_3) \in \Pred{\expand{\W^i}}\) such that \(K_2
  \subsetneq K_3\). We can repeat this
  again and again. Assume we never find a maximal element:
  we have the infinite sequence
  \(K_1 \subsetneq K_2 \subsetneq \dotsb\) such that each \((v, K_j) \in
  \Pred{\expand{\W^i}} \setminus \reduce{\Pred{\expand{\W^i}}}\).

  By definition of \(\Pred{\expand{\W^i}}\),
  for each \(j \in \natnonneg\), there exists \(a_j \in \En(v)\) such that,
  for any \(v' \in \vertices\), \((v', K_j \cap \opponents{v, a_j,
  v'}) \in \expand{\W^i}\).
  As \(\En(v)\) is finite, there are infinitely many \(\ell < \ell'\) such that
  \(a_\ell = a_{\ell'}\) and \(K_\ell \subsetneq K_{\ell'}\).
  So, we can group together all \(K_j\) that depend on
  the same action.
  Without loss of generality, we can assume that the
  sequence of \(K_j\) for a common action is maximal (with regard
  to \(\subseteq\)), \ie, assume that
  \(\exists a \in \En(v),
    \forall j \in \natnonneg,
    \forall v' \in \vertices :
  (v', K_j \cap \opponents{v, a, v'}) \in \expand{\W^i}\), which is
  equivalent to
  \[
    \exists a \in \En(v),
    \forall v' \in \vertices,
    \forall j \in \natnonneg :
    (v', K_j \cap \opponents{v, a, v'}) \in \expand{\W^i}.
  \]

  Let \(v' \in \vertices\).
  As \((\subsets{\natnonneg}, \subseteq)\) is a complete lattice,
  there exists \(\mathfrak{K}_a\) that is the join of the
  maximal set
  \(\{K_1, K_2, \dotsc\}\), \ie, such that, for every \(j\),
  \(K_j \subseteq \mathfrak{K}_a\). For this complete
  lattice, the join is the union of the sets (see
  \Cref{ex:lattice,ex:lattice:complete}):
  \(\mathfrak{K}_a = \bigcup_{j \in \natnonneg} K_j\).
  Notice that picking a different action \(a'\) may yield
  a different \(\mathfrak{K}_{a'}\), as the choice of the action may
  lead to different knowledge sets. Observe that, since \(\{K_1,
  K_2 \dotsc\}\) is the largest possible set, \(\mathfrak{K}_a\) is
  the maximal set for the selected action \(a\). We argue that
  there is an action that gives us \(\mathfrak{K}_a \in
  \reduce{\Pred{\expand{\W^i}}}\).

  Observe that \(K_1 \cap \opponents{v, a,
  v'} \subseteq K_2 \cap \opponents{v, a, v'} \subseteq \dotsb\), as
  \(K_1 \subsetneq K_2 \subsetneq \dotsb\).
  Recall that \(\W^i = \reduce{\W^i} = \reduce{\expand{\W^i}}\), so,
  there exists \((v, \mathfrak{K}'_a) \in \W^i\) such that, for all \(j\),
  \((v', K_j \cap \opponents{v, a, v'}) \parOrder (v, \mathfrak{K}_a)\), \ie,
  \(\mathfrak{K}'_a\) is the join of the set
  \(\{K_1 \cap \opponents{v, a, v'}, \dotsc\}\).
  Again, in
  \((\subsets{\natnonneg}, \subseteq)\), the join is the union:
  \(\mathfrak{K}'_a = \bigcup_{j \in \natnonneg} K_j \cap
  \opponents{v, a, v'}\).

  By definitions, \(\mathfrak{K}_a \cap \opponents{v, a, v'} = \bigcup_{j \in
  \natnonneg} K_j \cap \opponents{v, a, v'} = \mathfrak{K}'_a\). Hence,
  \((v', \mathfrak{K}_a \cap \opponents{v, a, v'}) \in \expand{\W^i}\),
  \ie, \((v, \mathfrak{K}_a) \in \Pred{\expand{\W^i}}\).
  It may happen that the selected \(a\) does not yield a
  maximal element: indeed, another action \(a'\) may lead us to \((v,
  \mathfrak{K}_{a'}) \in \Pred{\expand{\W^i}}\) (using exactly the
  same arguments), but such that \(\mathfrak{K}_a \subseteq
  \mathfrak{K}_{a'}\). Since there are finitely many actions, there
  necessarily exists an action that yields a maximal element of
  \(\Pred{\expand{\W^i}}\), \ie, there exists \(\mathfrak{a} \in
  \alphabet\) such that \((v, \mathfrak{K}_\mathfrak{a}) \in
  \reduce{\Pred{\expand{\W^i}}}\).
\end{proof}

\begin{corollary}\label{cor:W:WNonDecreasing}
  For every \(i \in \nat\), \(\expand{\W^i} \subseteq
  \expand{\W^{i+1}}\).
\end{corollary}
\begin{proof}
  As \(W^i = \reduce{\W^i}\), \(\expand{\W^i} =
  \expand{\reduce{\W^i}}\).
  By \Cref{lemma:W:welldefined}, \(\expand{\Pred{\expand{\W^i}}} =
  \expand{\reduce{\Pred{\expand{\W^i}}}}\).
  Hence,
  \begin{align*}
    \expand{\W^{i+1}}
    &= \expand{(\W^i \lub \Pred{\expand{\W^i}})}
    = \expand{\reduce{\W^i \cup \Pred{\expand{\W^i}}}}
    & \text{by definition of \(\lub\)}
    \\
    &= \expand{(\reduce{\W^i} \cup \reduce{\Pred{\expand{\W^i}}})}
    & \text{by \Cref{prop:expand_is_expand_reduce}}
    \\
    &= \expand{\reduce{\W^i}} \cup \expand{\reduce{\Pred{\expand{\W^i}}}}
    & \text{by \Cref{prop:expand_is_expand_reduce}}
    \\
    &= \expand{\W^i} \cup \expand{\reduce{\Pred{\expand{\W^i}}}}
    & \text{as \(\expand{\W^i} = \expand{\reduce{\W^i}}\).}
  \end{align*}
  From there, it is clear that \(\expand{\W^i} \subseteq \expand{\W^{i+1}}\).
\end{proof}

From that result, the existence of \(\lim_{i \to +\infty}
\expand{\W^i} = \bigcup_{i \in \nat} \expand{\W^i}\) is immediate.
Nevertheless, it is unclear whether
$\lim_{i \to +\infty} \W^i$, denoted \(\W^\infty\) hereafter, exists.
To show that the sequence stabilizes, we will require
an alternative approach, which will be proved to terminate. We will
also argue that both compute the same sets at each iteration. Hence,
\(\W^\infty\) exists.
Nonetheless, we argue that the least
fixed-point of the \(\Pred{\cdot}\) operator of \Cref{def:fixpoint}
corresponds to the winning region for Eve in the knowledge game.

\begin{restatable}{theorem}{fixpointcorrectness}\label{thm:symbolic:strategy:region}
  Let \((v, K) \in \verticesE\) such that \(K \neq \emptyset\).  Eve
  has a winning strategy from \((v, K)\) in \(\kGame\) if and only if
  there exists $i \in \natnonneg$ such that
  \((v, K) \in \expand{\W^i}\).
\end{restatable}

\begin{proof}[Sketch]
  The proof follows the same line as the correctness proof of the
  fixed-point computation of the winning region in standard two-player
  turn-based reachability games.
  From every \((v, K) \in \expand{\W^i}\), one can construct a winning
  strategy by induction on the least index \(j\) such that
  \((v, K) \in \expand{\W^j}\): select an action such
  that any successor \((v', K')\) satisfies
  \((v', K') \in \expand{\W^{j-1}}\). Conversely, if Eve has a winning
  strategy from \((v, K)\), she can ensure to progress towards a
  target vertex.  Since any target vertex necessarily belongs to the
  winning region, we can show that \((v, K)\) is also in the winning
  region.
\end{proof}

From \Cref{thm:knowledge:strategy,thm:symbolic:strategy:region}, we
conclude that
\(\bigcup_{i \in \natnonneg} \expand{\W^i}\) characterizes the winning
region for Eve in the concurrent parameterized game \(\game\).

\begin{restatable}{corollary}{fixpointcorollary}
  Eve has a winning strategy from a vertex \(v_0\) in \(\game\) if and only if
  there exists $i \in \nat$ such that \((v_0, \natnonneg) \in
  \W^i\). Furthermore, if $\W^\infty$ is defined, then Eve has a
  winning strategy from \(v_0\) in \(\game\) if and only if
  \((v_0, \natnonneg) \in \W^\infty\).
\end{restatable}

\subsection{An alternative fixed-point approach}\label{subsec:alternative}
The previous fixed-point characterization cannot be immediately turned
into an algorithm, because, even if the parameterized arena is finite,
the operator $\KPred{\cdot,\cdot}{\cdot}$ collects a priori
unboundedly many sets $K$. Therefore, it is not obvious how to use
it to compute the winning region in the knowledge game. We thus
propose an alternative approach, with an operational viewpoint, to
compute \(\W^{i+1}\) from \(\W^i\).  Let \(v \in \vertices\) and
\(a \in \En(v)\).  Before providing the definition, we give the
intuition behind it:
\begin{itemize}
  \item
    Fix, for each \(v' \in V\), a knowledge set \(K_{v'}\) such that
    \((v', K_{v'}) \in \W^i\).
  \item We want to focus on the vertices \(v'\) that can be reached from
    \(v\) by taking an edge labeled with \(a\) and some constraint on
    the number of opponents.  Since the arena may be non-deterministic,
    there can be multiple such vertices.  Hence, we iterate over each
    alternative, by taking every possible subset \(V'\) of \(V\),
    representing the vertices that we want to reach.

  \item We first compute the intersection of every \(K_{v'}\), with
    \(v' \in V'\).  This yields a set of numbers of opponents for which
    we are sure to have a winning strategy from each \(v'\).  Indeed,
    for every \(v' \in V'\), \((v', K_{v'}) \in \W^i\).  As
    \(\bigcap_{v' \in V'} K_{v'}\) is a subset of each \(K_{v'}\), we have
    \((v', \bigcap_{v' \in V'} K_{v'}) \in \expand{\W^i}\).

  \item It remains to check that we can guarantee to \emph{only} reach
    vertices in \( V'\), while avoiding any vertex \(v''\) with
    \(v'' \notin V'\).  That is, we want to reach
    $\opponents{v,a,V'}$, as defined in~\Cref{nablaI}
    page~\pageref{nablaI}.

  \item We take the intersection of
    \(\bigcap_{v' \in V'} K_{v'}\) and \(\opponents{v,a,V'}\), to obtain
    the numbers of opponents for which we know that we have a winning
    strategy from \(v\).

  \item We repeat this procedure for every possible combination of sets
    \((v', K_{v'}) \in \W^i\), by varying the witnesses
    ${(K_{v'})}_{v' \in V}$.
    Let us emphasize that we do not take the union over the
    possible \(K_{v'}\). Instead, we construct a set containing all
    the sets yielded by the above procedure.
\end{itemize}

\begin{definition}\label{def:symbolic-alt}
  For $(v, a) \in \vertices \times \alphabet$ and $S \subseteq
  \verticesE$, let
  \(\KPredAlt{v,a}{S} \) be the operator on \(\subsets{\natnonneg}\),
  and let \(\PredAlt{S}\) be the operator on \(\subsets{\verticesE}\)
  defined as
  \begin{gather*}
    \KPredAlt{v,a}{S} = \left\{ \bigcup_{V' \subseteq V}
      \left(\opponents{v,a,V'} \cap \bigcap_{v' \in V'} K_{v'}\right)
      \ \middle\vert\
    \forall v' : (v', K_{v'}) \in S \right\}
    \\
    \PredAlt{S} = \bigcup_{v\in \vertices} \bigcup_{a\in
    \En(v)}\{(v,K) \mid K \in \KPredAlt{v,a}{S}\} \enspace.
  \end{gather*}
  Finally, the sets \({(\WAlt^i)}_{i \in \nat} \subseteq \verticesE\)
  are defined inductively:
  \begin{equation*}
    {\displaystyle \WAlt^0} {\displaystyle = \{(\target, \natnonneg)\} \cup
    \{(v, \emptyset) \mid v \in \vertices \setminus \{\target\}\}}
    \qquad\text{and}\qquad
    {\displaystyle \forall i \geq 0 : \WAlt^{i+1}} {\displaystyle =
      \WAlt^i \lub
    \PredAlt{\WAlt^i}.}
  \end{equation*}
\end{definition}

For the \(\PredAlt{\cdot}\) operator, one can directly show termination.

\begin{restatable}{lemma}{predAltTermination}\label{lemma:PredAlt:termination}
  The sequence \({(\WAlt^i)}_{i \in \nat}\) stabilizes.
\end{restatable}
\begin{proof}
  Let \(\latticeK\) be the minimal set such that
  \begin{itemize}
    \item \(\natnonneg \in \latticeK\), \(\emptyset \in \latticeK\),
    \item for every \(v, v' \in \vertices\) and \(a \in \alphabet\),
      \(\opponents{v, a, v'} \in \latticeK\), and
    \item for every \(K, K' \in \latticeK\), one has
      \(K \cup K' \in \latticeK\), \(K \cap K' \in \latticeK\), and
      \(K \setminus K' \in \latticeK\),
  \end{itemize}
  and let \(\verticesEK = \{(v, K) \in \verticesE \mid K \in \latticeK\}\).
  One can show that \(\latticeK\) and \(\verticesEK\) are finite
  (using algebraic arguments and the fact that we have
  finitely many \enquote{atoms}).

  Whenever $S$ is a subset of $\verticesEK$, then
  $\KPredAlt{v,a}{S}$ belongs to $\latticeK$ and $\PredAlt{S}$ is included
  in $\verticesEK$. $\verticesEK$ is thus stable by the operator. Since
  $\WAlt^0 \subseteq \verticesEK$, the sets $\WAlt^i$ all belong to
  $\verticesEK$.  As $\latticeK$ is finite, the fixed-point
  computation of the sequence ${(\WAlt^i)}_{i \in \nat}$ will thus
  converge in finitely many steps.
\end{proof}
Note that proving that \(\KPred{v, a}{S}\) belongs to
\(\latticeK\) is not an easy feat, and will be obtained as a
consequence of \(\W\) and \(\WAlt\) computing the same sets. This
justifies the introduction of \(\KPredAlt{\cdot}{\cdot}\).

We can also show a monotonicity property with respect to
the preorder.

\begin{restatable}{lemma}{predAltMonotonicity}\label{lemma:PredAlt}
  For any $(v,a) \in \vertices \times \alphabet$ and two sets
  $S,S' \subseteq \verticesE$ such that $S \parOrdersim S'$,
  \begin{itemize}
    \item
      \(\KPredAlt{v,a}{S} \subseteq \expand{(\KPredAlt{v,a}{S'})}\); and
    \item
      \(\PredAlt{S} \parOrdersim \PredAlt{S'}\).
  \end{itemize}
\end{restatable}

\begin{example}\label{ex:WAlt}
  We illustrate this alternative approach on the game
  \(\game = (\arena, t)\) of
  \Cref{fig:arena:nondeterministic}. Initially
  $\WAlt^0 = \{(s, \emptyset), (v, \emptyset), (x_1, \emptyset), (x_2,
  \emptyset), (y_1, \emptyset), (y_2, \emptyset), (t, \natnonneg)\}$.

  We focus on computing \(\KPredAlt{v, c}{\WAlt^0}\).
  Given \(\WAlt^0\), there is a unique choice
  for \(K_{v'}\) for each \(v' \in \vertices\). Now, since the only
  successors of $v$ by a $c$-edge are $s$ and
  $t$, for every \(V' \subseteq \vertices\),
  if \(\opponents{v, c, V'}\) is not empty, then it must be that
  $V' \subseteq \{s,t\}$.
  When \(s \in V'\), we immediately conclude that
  \(\bigcap_{v' \in V'} K_{v'} = \emptyset\).  It thus remains to
  treat \(V' = \{t\}\), which yields
  \(\opponents{v, c, t} \cap K_t = \natsubs{>2}\).  Hence,
  \(\KPredAlt{v, c}{\WAlt^0} = \{\emptyset \cup \natsubs{>2}\} =
  \{\natsubs{>2}\}\).
  Analogously, one can compute
  \(\KPredAlt{y_1, b}{\WAlt^0} = \KPredAlt{y_2, b}{\WAlt^0} =
  \{\natnonneg\}\). So,
  \[
    \WAlt^1 = \{(s, \emptyset), (v, \natsubs{>2}), (x_1, \emptyset),
      (x_2, \emptyset),
    (y_1, \natnonneg), (y_2, \natnonneg), (t, \natnonneg)\}.
  \]
  Similarly, we obtain that
  \(\KPredAlt{x_1, b}{\WAlt^1} = \{\natsubs{\neq 2}\}\) and
  \(\KPredAlt{x_2, b}{\WAlt^1} = \{\natsubs{\neq 1}\}\).
  Thus,
  \[
    \WAlt^2 = \{(s, \emptyset), (v, \natsubs{>2}), (x_1, \natsubs{\neq
      2}), (x_2, \natsubs{\neq 1}), (y_1, \natnonneg), (y_2,
    \natnonneg), (t, \natnonneg)\}.
  \]
  Observe that $x_1$ and $x_2$ are the only targets of an edge labeled $a$
  from $v$. To compute
  \(\KPredAlt{v, a}{\WAlt^2}\), we focus on  \(\{x_1\}, \{x_2\}\),
  and \(\{x_1, x_2\}\).
  Due to the non-determinism, we are not sure to land exactly in
  \(x_1\) or in \(x_2\):
  \(\opponents{v, a, \{x_1\}} = \opponents{v, a, x_1}
    \setminus \bigcup_{v'
    \neq x_1} \opponents{v, a, v'} = 
  \emptyset\).  Likewise, \(\opponents{v,
  a, \{x_2\}} = \emptyset\).
  Finally, \(\opponents{v, a, \{x_1, x_2\}} = \natnonneg\), and
  \(\bigcap_{v' \in \{x_1,x_2\}} K_{v'} = \nat_{\neq 1} \cap \nat_{\neq 2}
  = \nat_{> 2}\).
  Thus, \(\KPredAlt{v, a}{\WAlt^2} = \nat_{> 2}\), but \((v, \nat_{>
  2})\) is already in \(\WAlt^2\). So, \(\WAlt^3 = \WAlt^2\).

  Observe that, for every index \(i\), \(\WAlt^i = \W^i\) (with
  \(\W^i\) computed in \Cref{ex:WK}). This is not a coincidence, and
  we prove it in the next section.
\end{example}

\subsection{Two fixed-point
computations?}\label{subsec:twofixedpointcomputations}
As we shall see, the two previous approaches contain exactly the
same sets.

\begin{restatable}{theorem}{fixpointequality}\label{thm:same-fixed-point}
  For all \(i \in \nat\), \(\W^i = \WAlt^i\). \\ The sequences
  ${(\W^i)}_{i \in \nat}$ and ${(\WAlt^i)}_{i \in \nat}$ stabilize, and
  \( \W^\infty = \WAlt^\infty\).
\end{restatable}

\begin{proof}[Sketch]
  We give here the main ingredients of the proof, starting with a strong
  relationship between
  $\KPred{\cdot,\cdot}{\cdot}$ and $\KPredAlt{\cdot,\cdot}{\cdot}$:
  for every \(v \in \vertices\), \(a \in \En(v)\) and
  $S \subseteq \verticesE$, we have
  \(\KPred{v, a}{S} \subseteq \KPredAlt{v, a}{S}\)
  and
  \(\KPredAlt{v,a}{S} \subseteq \KPred{v,a}{\expand{S}}\).

  Notice that the second inclusion requires taking the downward
  closure of $S$. Consider for instance an edge
  $v \xrightarrow{a,[2,5]} t$ with the assumption that
  $(t,\natnonneg) \in S$. Then $K_t = \natnonneg$ and
  $\opponents{v,a,t} \cap K_t = [2,5]$, whereas $(t,[2,5]) \notin
  S$. This also justifies the downward closure in the definition of
  ${(\W^i)}_{i \in \nat}$.

  If \(S\) is downward-closed, we derive
  \(\KPred{v, a}{S} = \KPredAlt{v, a}{S}\).
  From there, it should be clear that
  \(\Pred{S} = \PredAlt{S}\).

  Clearly, $\WAlt^0 = \W^0$.  Unfortunately, \(\W^0\)
  is not downward-closed: it contains
  $(t,\natnonneg)$ and none of its dominated elements. However,
  we still manage to prove by induction on $i$ that for every
  $i\in \nat$, $\WAlt^i = \W^i$:
  \(\WAlt^{i+1} = \WAlt^i \lub \PredAlt{\WAlt^i}
    = \W^i \lub \PredAlt{\expand{\WAlt^i}}
    =\W^i \lub \Pred{\expand{\W^i}}
  = \W^{i+1}\).

  By \Cref{lemma:PredAlt:termination}, the sequence
  \({(\WAlt^i)}_{i \in \nat}\) stabilizes, and so does
  \({(\W^i)}_{i \in \nat}\) thanks to the equality.
\end{proof}

Combining Theorems~\ref{thm:symbolic:strategy:region}
and~\ref{thm:same-fixed-point} we obtain our final result:
\begin{theorem}\label{thm:main}
  The antichains $\W^\infty$ and $\WAlt^\infty$ describe
  the maximal elements of the winning region for Eve in $\kGame$.
\end{theorem}

\section{Experimental evaluation}\label{sec:experiments}

We implemented our new algorithms \(\WAlt\) and \(\W\) based on
antichains, and the DFS algorithm of~\cite{BertrandBM19} in an
open-source C++ tool called
\textsf{ParaGraphs}\footnote{\url{https://gitlab.inria.fr/x-staquet/paragraphs}}.
It currently supports constraints given as finite unions of intervals,
but is designed in a modular way and makes it easy to define new
types (for instance semi-linear sets), as well as specific
implementations of the algorithms for these types.

\paragraph{Implementation of \(\W\).}

Recall that \(\KPred{\cdot, \cdot}{\cdot}\) requires to test every subset of
\(\natnonneg\), \ie, infinitely many sets.
However, as introduced in the proof of \Cref{lemma:PredAlt:termination}, any
set produced by \(\KPredAlt{\cdot, \cdot}{\cdot}\) belongs to a finite
subset \(\latticeK\) of \(\subsets{\natnonneg}\), which is a lattice, that
can be built from the constraints of the arena.
Moreover, by \Cref{thm:main}, for any \(i\), \(\W^i = \WAlt^i\).
Thus, it suffices to only consider sets that belong to \(\latticeK\).
That is, we can test finitely many sets.

In \textsf{ParaGraphs}, \(\latticeK\) is constructed from the constraints
appearing on the arena such that we can iterate over the elements, starting from
\(\natnonneg\), in an efficient way. This renders the implementation of
\(\W\) straightforward: to compute \(\KPred{v, a}{S}\), we iterate over
\(\latticeK\) and test whether the current element \(K\) is such that
\((v', K \cap \opponents{v, a, v'}) \notin S\), in which case we consider the
children of \(K\) in the lattice. Otherwise, \(K\) is a maximal element of the
set to return, and it is useless to explore the sublattice rooted at \(K\).
Furthermore, we keep track of which nodes can be ignored (because
they are already covered by a maximal element), in order to
avoid useless computations. Finally, we highlight that there is no
need to compute \(\expand{\W^i}\) to obtain \(\W^{i+1}\).
Indeed, we are only interested in knowing whether
\((v', K \cap \opponents{v, a, v'}) \in S\), which can be performed
directly on the maximal elements of the antichain.

\paragraph{Implementation of \(\WAlt\).}

Given its definition, the implementation of \(\WAlt\) is even more
direct, although heavily
optimized.
For instance, an antichain for \(\WAlt\) is represented as a
dynamically sized array;
to reduce time spent allocating and freeing it, the same array is reused
at each call of \(\KPredAlt{\cdot, \cdot}{\cdot}\).
Importantly, \(\WAlt\) does not require the computation of \(\latticeK\).

\paragraph{Implementation of the DFS algorithm.}

The DFS algorithm in~\cite{BertrandBM19} does not compute the
whole winning region, but determines whether Eve has a winning
strategy in \(\kGame\) from a given initial vertex
\((v_0, \natnonneg)\).  It performs a decomposition of the knowledge
game into sub-games, each of which is defined for a pair \((v, K)\)
and is a restriction of \(\kGame\) to vertices \((v', K, a)\) and
\((v', K')\) that are reachable from \((v, K)\) via vertices with exactly the
same knowledge set \(K\).  That is, except in the \enquote{exit
vertices}, the knowledge set is not refined within a sub-game.
Each sub-game can be solved in polynomial time in the
size of \(\game\). The algorithm tags each node of an exponential-size
tree \(\tree\) of sub-games with \Win or \Lose in a DFS fashion. This
DFS exploration is implemented in
\textsf{ParaGraphs}, together with the resolution of each sub-game by
standard (optimized) attractor computation.

\paragraph{Research statements}

Our experimental evaluation aims to address the following research
questions:
\begin{description}
  \item[Q1]
    Between \(\W\) and \(\WAlt\), which approach is the most efficient?
  \item[Q2]
    Are there arenas on which the antichain-based methods are more
    efficient than the DFS algorithm?
\end{description}

\subsection{Experimental results}

The experiments were carried out on two types of parameterized arenas:
on the one hand, scalable games derived from an illustrative example
in~\cite{BertrandBM19}, and on the other hand games obtained from
quantified Boolean formulas (QBF). Benchmarks were executed on a Dell
Precision 7780, running Debian 13.  \textsf{ParaGraphs} was compiled
with g++ 14.2.0.

\paragraph{Synthetic arenas.} To exercise
our tool \textsf{ParaGraphs}, we considered four families of scalable
instances, to cover a wide range of
cases. Each family has a parametric size: the parameter
is $n$, and the arena size (precisely the number of vertices and of
different constraints) is linear in $n$; some are deterministic
(D), some non-deterministic (ND), some are winning for Eve (W), some
are losing for Eve (NW).

\Cref{tab:experiments:solve:games} compares the running time of the three
algorithms, with and
without the time needed to construct the lattice for \(\W\).
To be more specific, the instances with
\(n = 11\) (resp.\ \(n = 12\)) generate a lattice consisting of 4096
(8192) sets that is constructed in about 15s (70s).
Computing larger
lattices exceeds our timeout of five
minutes.  We again highlight that the algorithm from~\cite{BertrandBM19}
solely determines whether Eve has a winning strategy, while our new
symbolic algorithms compute the whole winning region, thus returning a
more complete information.

\begin{table}[t]
  \caption{Time and memory taken by the three algorithms.  Timeout is
    set to five
    minutes. For \(\W\), we list time with and without the construction of the
    lattice, which counts towards the time limit.  The best time and
    memory consumption of each
  row is highlighted in bold.}\label{tab:experiments:solve:games}
  \centering
  \renewcommand{\multirowsetup}{\raggedleft}
  \renewcommand{\arraystretch}{0.8}
  \footnotesize
  \begin{tabular}{@{}
      l@{}r
      @{\hspace{10pt}}r@{\hspace{10pt}}r
      @{\hspace{10pt}}r
      @{\hspace{10pt}}r
      @{\hspace{10pt}}r@{\hspace{10pt}}r@{\hspace{10pt}}r
    @{}}
    \toprule
    \multicolumn{2}{c}{Game} &
    \multicolumn{4}{c}{Time (ms)} &
    \multicolumn{3}{c}{Memory (kB)}
    \\
    \cmidrule(r){1-2}
    \cmidrule(rl){3-6}
    \cmidrule(l){7-9}
    Name & Value of \(n\) &
    \multicolumn{2}{c}{\(\W\)} &
    \multirow{2}*{\(\WAlt\)} &
    \multirow{2}*{DFS} &
    \multirow{2}*{\(\W\)} &
    \multirow{2}*{\(\WAlt\)} &
    \multirow{2}*{DFS}
    \\
    \cmidrule(rl){3-4}
    & & Only & Total
    \\
    \midrule
    \multirow[c]{6}{*}{D-NW-1} & 10 & 850 & 3784 &
    \bfseries 67 & 108704 &
    7005.60 & 5067.20 & \bfseries 4642.40 \\
    & 11 & 3581 & 17312 & \bfseries 246 & timeout &
    10365.60 & \bfseries 5744.80 & timeout \\
    & 12 & 16017 & 82165 & \bfseries 961 & timeout &
    17808.80 & \bfseries 7768.00 & timeout \\
    & 13 & timeout & timeout & \bfseries 3817 & timeout &
    timeout & \bfseries 11936.80 & timeout \\
    & 14 & timeout & timeout & \bfseries 15269 & timeout &
    timeout & \bfseries 20584.00 & timeout \\
    & 15 & timeout & timeout & \bfseries 66657 & timeout &
    timeout & \bfseries 40322.40 & timeout \\
    \midrule
    \multirow[c]{6}{*}{D-W-1} & 10 & 608 & 3581 & 68 &
    \bfseries <1 &
    7491.60 & 5102.40 & \bfseries 4735.60 \\
    & 11 & 2559 & 16696 & 250 & \bfseries <1 &
    11694.40 & 6116.40 & \bfseries 4686.80 \\
    & 12 & 11493 & 80720 & 975 & \bfseries <1 &
    19790.80 & 8058.00 & \bfseries 4708.00 \\
    & 13 & timeout & timeout & 3707 & \bfseries <1 &
    timeout & 12733.60 & \bfseries 4674.00 \\
    & 14 & timeout & timeout & 15105 & \bfseries <1 &
    timeout & 22339.60 & \bfseries 4738.80 \\
    & 15 & timeout & timeout & 66224 & \bfseries <1 &
    timeout & 43618.40 & \bfseries 4696.40 \\
    \midrule
    \multirow[c]{6}{*}{D-NW-2} & 10 & 211 & 3345 & 30 &
    \bfseries 2 &
    6391.20 & 4568.80 & \bfseries 4538.40 \\
    & 11 & 747 & 15574 & 85 & \bfseries 2  &
    8881.60 & 4804.80 & \bfseries 4716.00\\
    & 12 & 2720 & 74282 & 276 & \bfseries 3  &
    13774.00 & 5253.20 & \bfseries 4695.60\\
    & 13 & timeout & timeout & 945 & \bfseries 30 &
    timeout & 6324.00 & \bfseries 4541.60 \\
    & 14 & timeout & timeout & 3571 & \bfseries 41 &
    timeout & 8076.40 & \bfseries 4614.80 \\
    & 15 & timeout & timeout & 12425 & \bfseries 54 &
    timeout & 12889.60 & \bfseries 4616.00 \\
    \midrule
    \multirow[c]{6}{*}{D-W-2} & 10 & 473 & 3569 & 104 &
    \bfseries <1 &
    6517.20 & 4803.60 & \bfseries 4540.00 \\
    & 11 & 1764 & 16680 & 323 & \bfseries <1 &
    9087.60 & 5216.80 & \bfseries 4725.20 \\
    & 12 & 6668 & 78194 & 1065 & \bfseries 2 &
    14324.00 & 5982.40 & \bfseries 4673.20 \\
    & 13 & timeout & timeout & 3903 & \bfseries 13 &
    timeout & 7767.20 & \bfseries 4728.40 \\
    & 14 & timeout & timeout & 14157 & \bfseries 12 &
    timeout & 11820.40 & \bfseries 4718.80 \\
    & 15 & timeout & timeout & 53060 & \bfseries 13 &
    timeout & 19710.80 & \bfseries 4651.60 \\
    \midrule
    \multirow[c]{6}{*}{ND-NW} & 10 & 9 & 2935 &
    \bfseries <1 & 102195 &
    6830.00 & 4666.40 & \bfseries 4400.00 \\
    & 11 & 23 & 14230 & \bfseries 2 & timeout &
    10323.60 & \bfseries 4777.60 & timeout \\
    & 12 & 56 & 67943 & \bfseries 6 & timeout &
    16984.40 & \bfseries 5126.80 & timeout \\
    & 13 & timeout & timeout & \bfseries 9 & timeout &
    timeout & \bfseries 5778.00 & timeout \\
    & 14 & timeout & timeout & \bfseries 20 & timeout &
    timeout & \bfseries 7258.40 & timeout \\
    & 15 & timeout & timeout & \bfseries 40 & timeout &
    timeout & \bfseries 10236.80 & timeout \\
    \bottomrule
  \end{tabular}
\end{table}

Based on these results, we can answer our two research
questions for synthetic arenas. More specifically for {\bf Q1},
\(\WAlt\) is more efficient than \(\W\), as
observed by the number of timeouts. This difference majorly comes from
the construction of \(\latticeK\): computing the lattice in advance
already costs more time than the iterations of \(\WAlt\).

Still for synthetic arenas, the answer to {\bf Q2} is
\enquote{yes}, as, on D-NW-1, \(\WAlt\) greatly outperforms the DFS
algorithm of~\cite{BertrandBM19}.  Interestingly, the latter vastly
outperforms the former when Eve has a winning strategy (D-W-1 and
D-W-2) and for D-NW-2. For the two winning cases, this is simply due
to the fact that the tree exploration is much quicker:
since there is a winning strategy, we explore less sub-games than
compared to D-NW-1 and D-NW-2.  For instance, for \(n = 10\), 31
sub-games are considered for D-W-1, and 19728201 for D-NW-1. Moreover,
the structure of D-NW-2 already induces a small tree: only 135
sub-games are needed (and 34 for D-W-2). Both \(\W\) and \(\WAlt\)
need 22 iterations to compute the whole winning
region.\footnote{Recall that, for every \(i\), \(\WAlt^i = \W^i\), \ie, it is
expected that both algorithms need the same number of iterations.}
Finally, notice that \(\WAlt\) is the only algorithm that managed to
finish all instances. Thus, it seems to scale well to larger and more
complex arenas.

We measured the memory consumption of the algorithms,
which is also reported in \Cref{tab:experiments:solve:games}.
We observe that D-W-1 is
the arena that induces the highest memory consumption for every algorithm.
The fact that the antichain-based methods require more
space than the DFS algorithm is not surprising. Indeed, one may
need to store exponentially many maximal elements inside an
antichain, \ie, our new algorithms may require exponential space in
the worst case.

\begin{figure}[t]
  \centering
  \begin{tikzpicture}[
      /pgfplots/table/x expr={(\thisrow{Variables}==7 &&
      \thisrow{Clauses}>=60)?\thisrow{Clauses}:nan},
    ]
    \begin{axis}[
        footnotesize,
        width = 1\textwidth,
        height = 90pt,
        axis x line = bottom,
        axis y line = left,
        xlabel = {Number of clauses},
        xlabel near ticks,
        ylabel = {Time (ms)},
        ylabel near ticks,
        ylabel shift = {-5pt},
        ymajorgrids = true,
        xtick = data,
        enlarge y limits = 0.1,
        enlarge x limits = 0.1,
        nodes near coords={\pgfmathprintnumber[1000
        sep={\,}]\pgfplotspointmeta},
      ]
      \addplot+ [
        Red,
        only marks,
        mark=*,
      ]
      table [
        y expr = \thisrow{W time total (ms)},
        col sep = comma,
      ]
      {qbfSolving.csv}
      ;
      \addplot+ [
        Blue,
        only marks,
        mark=triangle*,
      ]
      table [
        y expr = \thisrow{WAlt time (ms)},
        col sep = comma,
      ]
      {qbfSolving.csv}
      ;
      \addplot+ [
        Indigo,
        only marks,
        mark=square*,
      ]
      table [
        y expr = \thisrow{Explicit time (ms)},
        col sep = comma,
      ]
      {qbfSolving.csv}
      ;
    \end{axis}
  \end{tikzpicture}
  \caption{Time taken by \(\W\), \(\WAlt\) on randomly generated QBF
    instances with seven variables. For each number of
    clauses, we generate 50 instances, and we compute the average,
    ignoring the runs that timed out. The timeout was set to five
    minutes. On those instances, the DFS
    algorithm from~\cite{BertrandBM19} systematically reached the
    timeout limit, while \(\W\) and \(\WAlt\)
    always succeed. Red circles are for \(\W\) and blue triangles for
    \(\WAlt\). Values are displayed  next to the
  symbol.}\label{fig:experiments:solve:qbf}

  \begin{tikzpicture}[
      /pgfplots/table/x expr={(\thisrow{Variables}==7 &&
      \thisrow{Clauses}>=60)?\thisrow{Clauses}:nan},
    ]
    \begin{axis}[
        footnotesize,
        width = 1\textwidth,
        height = 85pt,
        axis x line = bottom,
        axis y line = left,
        xlabel = {Number of clauses},
        xlabel near ticks,
        ylabel = {Memory (MB)},
        ylabel near ticks,
        ylabel shift = {-5pt},
        ymajorgrids = true,
        xtick = data,
        enlarge y limits = 0.1,
        enlarge x limits = 0.1,
        nodes near coords={\pgfmathprintnumber[1000
        sep={\,}]\pgfplotspointmeta},
      ]
      \addplot+ [
        Red,
        only marks,
        mark=*,
      ]
      table [
        y expr = \thisrow{W memory (kB)} / 1000,
        col sep = comma,
      ]
      {qbfMemory.csv}
      ;
      \addplot+ [
        Blue,
        only marks,
        mark=triangle*,
      ]
      table [
        y expr = \thisrow{WAlt memory (kB)} / 1000,
        col sep = comma,
      ]
      {qbfMemory.csv}
      ;
      \addplot+ [
        Indigo,
        only marks,
        mark=square*,
      ]
      table [
        y expr = \thisrow{Explicit memory (kB)} / 1000,
        col sep = comma,
      ]
      {qbfMemory.csv}
      ;
    \end{axis}
  \end{tikzpicture}
  \caption{Memory taken by \(\W\), \(\WAlt\), and the DFS algorithm on
    randomly generated QBF instances, with seven variables. For each number of
    clauses, we generate 50 instances, and we compute the average,
    ignoring the runs that timed out. All executions of the DFS
    algorithm on those instances timed
  out.}\label{fig:experiments:solve:qbf:memory}
\end{figure}

\paragraph{Arenas obtained from quantified Boolean formulas.} Our
tool \textsf{ParaGraphs} also handles quantified Boolean formulas
(QBF) and checks their satisfiability by reduction to solving a
concurrent parameterized reachability game, as in the \PSPACE-hardness
proof~\cite[Prop.\ 13]{BertrandBM19}.

Interestingly, given the constraints in arenas generated from QBF
instances, the constructing the lattice \(\latticeK\) (whose size only
depends on the number of variables) can be made very efficient: a
clever order on insertions  always avoids costly equality checks
(with elements already in the lattice). For formulas with 6 (resp.\ 7)
variables, the lattice, whose size is 8192 (resp.\ 32768), is obtained
in about 575ms (resp.\ 9s). Observe that our tailored construction
permits us to build larger lattices. We exercised the implementation
of the three algorithms on bigger arenas than the previous ones by
randomly generating 50 formulas with \(n\) variables and \(m\)
clauses, for every \(n \in \{4,5,6,7\}\) and
\(m \in \{50,60,\dotsc,100\}\). We focus here on the instances with 7
variables and defer to~\cite{BertrandBS25} the other instances.

The answers to {\bf Q1} and {\bf Q2} on benchmarks
generated from QBF formulas are similar to the ones on synthetic
arenas.  First, both \(\W\) and \(\WAlt\) manage to finish all
instances, while the DFS algorithm systematically reached
the time limit.  Second, \Cref{fig:experiments:solve:qbf} gives the
time taken by the \(\W\) and \(\WAlt\) algorithms on the formulas with
7 variables and with 60 to 100 clauses. For clarity, the exact values
are displayed next to each point. We observe that \(\WAlt\) is the
most efficient algorithm, by a margin of three orders of
magnitude. The memory consumption is given in
\Cref{fig:experiments:solve:qbf:memory}. Similarly to the synthetic
arenas, \(\WAlt\) needs less memory than \(\W\), and the algorithm
from~\cite{BertrandBM19} it the most efficient memory-wise. However,
we do not measure the memory consumption of the instances that timed
out.

\paragraph{Conclusion.}
We can give the following answers to our research
questions:
\begin{description}
  \item [A1]
    \(\WAlt\) is more efficient than \(\W\), both time- and memory-wise.
  \item [A2]
    There are arenas on which \(\WAlt\) outspeeds the DFS algorithm.
    The memory consumption of the latter is always lower than those
    of the former.
    However, our experiments are not sufficient to deduce syntactic
    characteristics
    of arenas on which \(\WAlt\) outperforms the DFS algorithm.
    Moreover, a potential optimization that would allow a more direct
    comparison between the DFS algorithm and \(\WAlt\) would be to
    stop computing the winning region as soon as \((v_0,
    \natnonneg)\) is encountered.
\end{description}

\section{Conclusion}\label{sec:conclusion}

This paper proposes a symbolic approach to compute the winning region
in parameterized concurrent reachability games. It is based on
antichains to represent sets of knowledges in an effective way. After
providing a fixed-point characterization of the winning region, we
gave two effective algorithms to compute the least fixed-point, one
that iterates over the elements of a lattice (which is finite when the
arena is finite), and one that constructs the knowledge sets on the
fly.  We implemented both methods, as well as the DFS algorithm
from~\cite{BertrandBM19}, in an open-source tool called
\textsf{ParaGraphs}, and performed benchmarks on multiple examples,
including arenas obtained from quantified Boolean formulas.

\bibliographystyle{eptcs} \bibliography{abbreviations,references}

\end{document}